\documentclass[twoside,leqno,11pt]{article}

\usepackage[letterpaper]{geometry}

\usepackage{graphicx} 
\usepackage{ltexpprt}
\usepackage{booktabs}
\usepackage{amsfonts}
\usepackage{amssymb}

\usepackage{amstext}
\usepackage{amsmath}
\usepackage{xspace}
\usepackage{theorem}
\usepackage{graphicx}
\usepackage{mdframed}
\usepackage{comment}
\usepackage{mathdots}
\usepackage{xcolor}
\usepackage{graphics}
\usepackage{colordvi}
\usepackage{color}

\usepackage{titlesec}

\titleformat{\section}
  {\normalfont\Large\bfseries}  
  {\thesection}{1em}{}          

\titleformat{\subsection}
  {\normalfont\large\bfseries}
  {\thesubsection}{1em}{}

\usepackage{paralist}
\usepackage{bm}
\usepackage{mathpazo}

\usepackage{float}
\usepackage{kbordermatrix}
\usepackage{thmtools, thm-restate}

\usepackage{tikz}
\usetikzlibrary{snakes}
\usepackage[boxruled,lined,linesnumbered]{algorithm2e}
\usepackage{todonotes}
\usepackage{subcaption}
\usepackage{enumitem}
\usetikzlibrary{calc}
\usepackage{lineno}
\usepackage[colorlinks=true, linkcolor=black, citecolor=blue, urlcolor=blue]{hyperref}
\usepackage[capitalise,nameinlink]{cleveref}
\crefname{figure}{Figure}{Figures}

\newtheorem{theorem}{Theorem}
\newtheorem{conjecture}{Conjecture}

\newtheorem{lemma}{Lemma}

\newtheorem{question}{Question}
\newtheorem{corollary}{Corollary}

\newtheorem{definition}{Definition}

\newtheorem{remark}{Remark}

\newenvironment{proof}[1]{\noindent{\bf Proof }#1}%
        {\hspace*{\fill}$\Box$\par\vspace{4mm}}

\usepackage{authblk}

\newcommand{\GR}{\textsc{Greedy}\xspace}

\newcommand{\aset}{{\mathcal A}}

\begin{document}

\title{$\GR$ BST on Permutation Initial Tree}

\author{
    Akash Pareek\\
    Department of Computer Science and Automation,\\
    IISc Bangalore\\
    
    akashpareek@iisc.ac.in
}

\date{}

\maketitle







\begin{abstract}

The $\GR$ binary search tree (BST) algorithm, like the Splay tree, is a prominent candidate for the \emph{dynamic optimality conjecture}. While $\GR$ satisfies many desirable properties of BST, its cost and analysis to execute a search sequence $S$ is known to depend heavily on the choice of the \emph{initial tree} configuration. Most prior analyses assume a flat (empty) initial tree, under which several tight bounds are established. 

In this work, we introduce the notion of a \emph{permutation initial tree}, a specific class of non-flat initial tree and prove that for any permutation search sequence $S=(s_1,s_2,\dots, s_n)$, there exists a permutation initial tree $I_p$ such that the cost of $\GR$ on $I_p$ is same as its cost on the flat initial tree.

As an application of our result, we show that the \emph{preorder traversal conjecture} holds for $\GR$ when the initial tree is a permutation initial tree. While it was previously known that $\GR$ achieves an $O(n)$ cost on preorder sequences for flat initial tree (Chalermsook et al., FOCS 2015), our result demonstrates that the same linear bound holds when the initial tree is a permutation initial tree. This result also matches the $O(n)$ bound for Splay tree on preorder sequence when the initial tree aligns with the traversal order (Chaudhuri and Höft, SIGACT 1993).

\end{abstract}
\newpage
\section{Introduction}\label{sec:intro}

Define the set $\{1,2,\dots,n\}$ as $[n]$, where each element in $[n]$ is referred to as a \emph{key}. Searching for a key is a fundamental operation in data structures, and one data structure that is widely used for this purpose is a binary search tree (BST). In a BST, each node contains a key and has pointers to its left and right children, as well as to its parent (except the root node, which has no parent pointer). Each node in a BST can have up to two children. Assuming all nodes have distinct keys, a BST has the following property: the keys in the left subtree are always less than the key of its parent node, and the keys in the right subtree are always greater than the key of its parent node.

Balanced BSTs such as AVL tree \cite{adel1962algorithm}, Red-Black tree \cite{guibas1978dichromatic} have worst-case time complexity of  $O(\log n)$ per search. Sleator and Tarjan introduced an online self-adjusting binary search tree named the \emph{Splay tree} \cite{SleatorT85}, which has an \emph{amortized cost} of $O(\log n)$ per search. In the same paper, the authors introduced the famous \emph{dynamic optimality conjecture}. Let $S=(s_1,s_2,\dots,s_m)$ be a search sequence. Let $OPT(S)$ denote the optimal offline cost to serve the sequence $S$ and let $Cost_{\aset}$(S) denote the cost of an online algorithm $\aset$ to serve the sequence $S$. The dynamic optimality conjecture states the following:

\begin{conjecture}(Dynamic Optimality Conjecture \cite{SleatorT85})\label{con:dynamic} For all search sequence $S$, $Cost_{\textsc{Splay tree}}(S)=O(OPT(S))$.
\end{conjecture}

The  dynamic optimality conjecture, originally defined for the Splay tree, can be extended to any online BST $\aset$ as follows:

\begin{definition} (Dynamically Optimal BST)
    A BST $\aset$ is said to be dynamically optimal if, for all search sequences $S$, $Cost_{\aset}(S)=O(OPT(S)).$
\end{definition}

 Apart from the Splay tree \cite{SleatorT85}, $\textsc{Greedy}$ BST
  (\cite{demaine_geometry,lucas,munro2000competitiveness}) is also conjectured to be dynamically optimal. The reason why Splay tree and $\textsc{Greedy}$ are considered as the candidates for dynamic optimality conjecture is that there are many properties of BST that are known to be satisfied by both the Splay tree and $\textsc{Greedy}$, such as \emph{balance theorem, static finger property, static optimality, working set theorem} \cite{SleatorT85,fox11}, \emph{dynamic finger} \cite{cole2000dynamic,IaconoL16}.

Despite these shared properties, one aspect that significantly affects the analysis of both the Splay tree and $\textsc{Greedy}$ is the choice of the \emph{initial tree}. An initial tree is a tree configuration present before executing a search sequence $S$ (see \cref{sec:initialtree} for details). If no such tree is specified, we assume an empty or \emph{flat initial tree}, denoted by $I_f$. Analyzing BST algorithms with a flat initial tree is often simpler, as the cost depends solely on the sequence $S$ and the BST algorithm. However, when the initial tree is not flat, both the execution cost and the analysis can change significantly. It is well known that the cost of $\textsc{Greedy}$ and Splay tree depends on the choice of the initial tree, with different initial trees leading to different costs for the same sequence $S$.
Hence, in this paper, we ask the following question for $\GR$.

\begin{question}
Does there exist an initial tree for which the cost of $\GR$ on a sequence $S$ matches its cost starting from the flat initial tree?
\end{question}

 We answer this question affirmatively. Specifically, for any permutation search sequence $S$, we show that there exists a corresponding \emph{permutation initial tree}, denoted by $I_p$, such that the cost of $\GR$ to execute the sequence $S$ on $I_f$ and $I_p$ is the same.


 Next, we use the above result to prove the \emph{preorder (traversal) conjecture} for a specific initial tree. To this end, we first describe the notion of \emph{pattern avoidance}, which will help us to define the \emph{preorder sequence} and the preorder (traversal) conjecture.
 
\vspace{0.4cm}

\textbf{Pattern avoidance:} For simplicity, we discuss only the permutation search sequence. Consider a permutation $ P = (a_1, a_2, \dots, a_n)$ and a permutation pattern  $P' = (b_1, b_2, \dots, b_l)$. We say that $P$ contains $P'$ if there exist indices  $i_1 < i_2 < \dots < i_l$ such that $a_{i_j} < a_{i_k}$ if and only if $b_j < b_k$ for all $j, k \in [l]$. In simpler terms, this means that $P$ possesses a subsequence that follows the same relative order as $P'$. In contrast, we say $P$ \emph{avoids} $P'$ if there does not exist a subsequence in $P$ that matches the relative order of $P'$. For example, $P=(50,40,30,20,10)$ does not contain the pattern $P'=(1,2)$, as no subsequence in $P$ is order isomorphic to $P'$.

We are now ready to define the preorder sequence and the preorder (traversal) conjecture.

\begin{definition} (Preorder Sequence) A sequence $S=(s_1,s_2,\dots,s_n)$ is called a preorder sequence if $S$ is obtained by a preorder traversal of a binary search tree $\aset'$ on $[n]$. Alternatively, $S$ is a preorder sequence if $S$ avoids the pattern $(2,3,1)$.
\end{definition}

\begin{conjecture}((Preorder) Traversal Conjecture)\label{con:preorder} Let $S$ be a preorder sequence. A BST algorithm $\aset$ is said to satisfy the traversal conjecture if starting with an arbitrary initial tree $I$, $Cost_{\aset}(S)=O(n)$.
\end{conjecture}


For Splay trees, in \cite{chaudhuri1993splaying}, the authors showed that when the sequence $S$ has the same relative order as the initial tree, then the Splay tree satisfies the traversal conjecture. Levy and Tarjan \cite{levy2019splaying} showed that inserting a preorder sequence $S$ into an empty tree via splaying takes $O(n)$ time. For an arbitrary initial tree, it is not known whether the preorder sequence is $o(n\log n)$.

For $\textsc{Greedy}$ on preorder sequences, the bounds are much better than the Splay tree. For preorder sequences with a flat initial tree, $\textsc{Greedy}$ is known to be $O(n)$ \cite{chalermsook2015pattern}. For arbitrary initial tree, in \cite{chalermsook2015pattern}, the authors gave a bound of $n2^{\alpha(n)^{O(1)}}$\footnote{$\alpha(n)$ is the inverse Ackermann function.} for any preorder sequence. This was recently improved to $O(n2^{\alpha(n)})$ \cite{chalermsook2023improved}.

\subsection{Our Results}

We begin by investigating the cost of $\textsc{Greedy}$ on the permutation initial tree, which matches the cost of $\textsc{Greedy}$ when the initial tree is flat. We show the following main theorem:

\begin{theorem}\label{lem:informalsamecost}(Informal)
    Let $S$ be a permutation search sequence, and let $I_p$ be a permutation initial tree, then $\textsc{Greedy}_{I_p}(S)=\textsc{Greedy}_{I_f}(S)$.
\end{theorem}

In the above theorem we show that for every permutation search sequence \( S \), there exists an initial tree, different from the flat initial tree, such that the cost of $\textsc{Greedy}$ on \( S \), denoted $\textsc{Greedy}(S)$, is same for both the initial trees. This is the first result to show that $\textsc{Greedy}(S)$ is the same for two different initial trees. 

Next, to demonstrate the application of  \cref{lem:informalsamecost}, we apply it to the preorder sequence, leveraging the known bound of preorder sequences when the initial tree is flat. By applying \cref{lem:informalsamecost} and the fact that preorder sequences have an \( O(n) \) cost on the flat initial tree for $\textsc{Greedy}$ \cite{chalermsook2015pattern}, we obtain the following theorem.

\begin{theorem}\label{thm:informalpreorder}(Informal)
  Let $S$ be a preorder sequence, then $\textsc{Greedy}_{I_p}(S)=O(n)$.
\end{theorem}

For the above theorem, we show that for every preorder sequence \( S \), there exists an initial tree, distinct from the flat initial tree, such that $\textsc{Greedy}(S) = O(n)$. This is the first result to show that $\textsc{Greedy}$ has a linear cost on preorder sequences when the initial tree is not flat. Notably, this matches the result of \cite{chaudhuri1993splaying} for Splay trees, thereby extending the parallel between the two algorithms which are candidates for dynamic optimality conjecture.

\subsection{Organisation}

    Section \ref{sec:prelims} covers the preliminaries, while Section \ref{sec:initialtree} discusses the concept of the initial tree. In Section \ref{sec:GRpermu}, we analyze $\textsc{Greedy}$ on a permutation initial tree and prove \cref{lem:informalsamecost}. Section \ref{sec:preorder} discusses the cost of preorder sequences on permutation initial tree and proves \cref{thm:informalpreorder}. Finally, Section \ref{sec:conclusion} concludes the paper and presents a few open problems.
 
\section{Preliminaries}\label{sec:prelims}

Consider a permutation search sequence $S = (s_1, s_2, \dots, s_n)$ where each $s_i$ is a key. This sequence can be visually represented as points in a plane. In this representation, the point $(s_i, i)$ denotes the search key $s_i$ at time $i$. Let $P$ denote the set of points corresponding to the sequence $S$. To simplify notation, we refer to both the search sequence and the set of points in the plane as $S$. For any point $p$ in the plane, $p.x$ represents its $x$-coordinate, which signifies the key, while $p.y$ represents its $y$-coordinate, indicating the time.

Consider any two distinct points $p$ and $q$ in the plane. If $p$ and $q$ do not lie on the same vertical or horizontal line, they define a rectangle denoted as $\square_{pq}$ if $p.x < q.x$, and otherwise denoted as $\square_{qp}$.

\begin{figure}
\begin{minipage}[b]{0.47\textwidth}
\centering
\begin{tikzpicture}[scale=0.47]
\draw[->] (0,-2) -- (10,-2);
\draw[->] (0,-2) -- (0,7);
\filldraw[gray!40!white,opacity=.5,draw=gray] (4,1) rectangle (6,3);
\filldraw[draw=black,fill=red!70!white] (3.9,0.9) rectangle (4.1,1.1) (5.9,2.9) rectangle (6.1,3.1);
\draw(4,0.75)node[anchor=north]{$p$} (6,3.25)node[anchor=south]{$q$};
\end{tikzpicture}
\end{minipage}
\hspace{.1cm}
\begin{minipage}[b]{0.427\textwidth}
\centering
\begin{tikzpicture}[scale=0.47]
\draw[->] (0,-2) -- (10,-2);
\draw[->] (0,-2) -- (0,7);
\filldraw[gray!40!white,opacity=.5,draw=gray] (4,1) rectangle (6,3);
\filldraw[draw=black,fill=red!70!white] (3.9,0.9) rectangle (4.1,1.1) (5.9,2.9) rectangle (6.1,3.1);
\filldraw[blue] (4,3) circle (3.5pt);
\draw(4,0.75)node[anchor=north]{$p$} (6,3.25)node[anchor=south]{$q$} (4,3.25)node[anchor=south]{$r$};
\end{tikzpicture}
\end{minipage}
\caption{(a) Arborally unsatisfied rectangle $\square_{pq}$. (b) Arborally satisfied rectangle $\square_{pq}$ due to point $r$.}
\label{fig:arborally}
\end{figure}
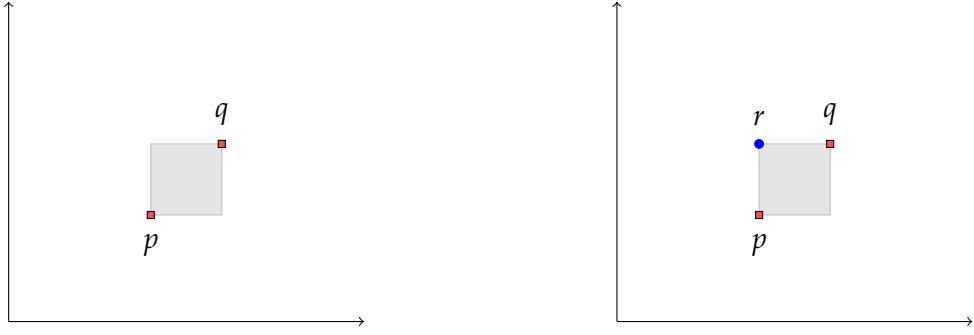

According to \cite{demaine_geometry}, a rectangle $\square_{pq}$ is termed \emph{arborally satisfied} in a point set $X$ if there exists another point $r \in X \setminus \{p, q\}$ such that $r$ lies in $\square_{pq}$ (See Figure 1(b)). Otherwise, the rectangle $\square_{pq}$ is said to be \emph{arborally unsatisfied} in $X$ (See Figure 1(a)). We now define an arborally satisfied set as follows:

\begin{definition}(Arborally satisfied set) A point set $X$ is said to be arborally satisfied if there exists another set $Z$ minimizing $|X \cup Z|$, where every pair of points in $X \cup Z$ is arborally satisfied.
\end{definition}

\cite{demaine_geometry} showed that \emph{finding the best BST execution for a sequence $S$ is equivalent to finding the minimum cardinality set $Z$ such that $S \cup Z$ is arborally satisfied}. In the online setting, where the point set \( S \) is revealed incrementally, an online arborally satisfied set can be constructed by adding the minimum number of points necessary to satisfy the arboral condition each time a new point is added to \( S \). To this end, we introduce the following Lemma from \cite{demaine_geometry}.

\begin{lemma}
\label{lem:bsttogeometry} (See Lemma 2.3 in \cite{demaine_geometry})
Let $A$ be an online algorithm that outputs an arborally satisfied set on any input representing $S$.  Then, there is an online BST algorithm $A'$ such that the cost of $A'$ is asymptotically equal to the cost of $A$, where the cost of $A$ is the number of points added by $A$ plus the size of $S$.
\end{lemma}

Demaine et al. \cite{demaine_geometry} gave an online BST algorithm named $\GR$ to find the arborally satisfied set for the search sequence $S$. Let us now describe the $\GR$ algorithm.

\vspace{0.4cm} 

\paragraph*{\textbf{$\textsc{Greedy}$ algorithm:}} For an online search sequence $S$, conduct a horizontal line sweep at each time $i$. At time $i$, identify all rectangles with one endpoint at $s_i$ and the other endpoint at $z$, where $z$ is positioned below the sweep line. If the rectangle $\square_{s_iz}$ is not arborally satisfied, introduce a point at the corner of the rectangle on the sweep line to make the rectangle $\square_{s_iz}$ arborally satisfied at time $i$. See  \cref{fig:greedy} for the execution of $\textsc{Greedy}$.

\begin{figure}[hpt!]
\centering
\begin{tikzpicture}[scale=.70]

\draw[->][black] (0,0) -- (10,0);
\draw[->][black] (0,0) -- (0,8);
\filldraw[gray!40!white,opacity=.5,draw=gray] (4,1) rectangle (6,3);
\filldraw[red!40!white,opacity=.5,draw=red] (5,5) rectangle (4,3);
\filldraw[green!40!white,opacity=.5,draw=green] (5,5) rectangle (6,3);
\filldraw[blue!40!white,opacity=.5,draw=blue] (8,7) rectangle (6,5);
\filldraw[draw=black,fill=red!70!white] (3.9,0.9) rectangle (4.1,1.1) (5.9,2.9) rectangle (6.1,3.1) (4.9,4.9) rectangle (5.1,5.1) (7.9,6.9) rectangle (8.1,7.1)   ;

\filldraw[blue] (4,3) circle (3.5pt) (4,5) circle (3.5pt) (6,5) circle (3.5pt) (6,7) circle (3.5pt);

\draw[black](11,0.5)node[anchor=north]{Key} (0,9.2)node[anchor=north]{Time};
\end{tikzpicture}

\caption{Execution of $\textsc{Greedy}$. Red points denote the searched keys. Blue points denote the keys added (touched) by $\textsc{Greedy}$ at any time $i$ to make all the rectangles arborally satisfied at time $i$.}
    \label{fig:greedy}
\end{figure}
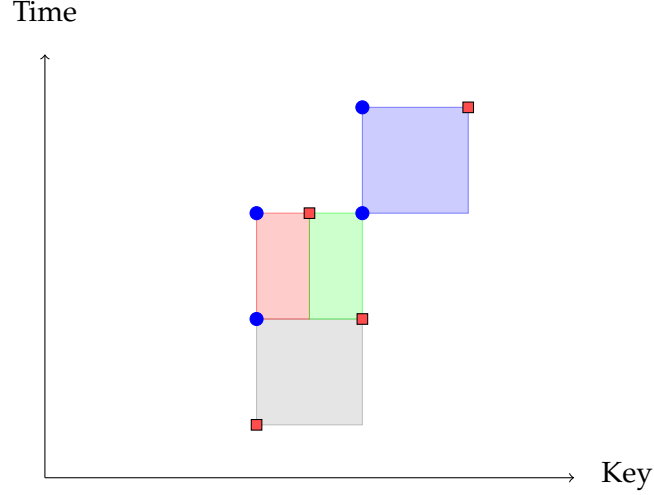

\begin{definition}(Touched key) A key $x$ is said to be touched at time $i$ if a point $p = (x,i)$  was added by $\textsc{Greedy}$  at time $i$. We also say that $\textsc{Greedy}$ added the point $p$ at time $i$.
\end{definition}

 The cost of the $\textsc{Greedy}$ algorithm for a permutation sequence $S$ is $n$ plus the number of touch points added by $\textsc{Greedy}$ for the execution of the sequence $S$. 

 Let $S$ be the set of points in the $x,y$ plane. We now define the mirror of a point set $S$, denoted as $\mathcal{M}(S)$ as follows:

\begin{definition}(Mirror point set)\label{def:mirror} The mirror point set of $S$, denoted as $\mathcal{M}(S)$, is the mirror image of the points in $S$ across the $x$-axis, i.e., a point $(s_i,i)\in S$ becomes $(s_i,-i)$ in $\mathcal{M}(S)$.
\end{definition}

For brevity, we also denote the mirror point set as a search sequence. Let $(s,t)\in \mathcal{M}(S)$, then we say that the key $s$ is searched at time $t$. With all the basic terminologies in hand, we now discuss the initial tree for BST.

\section{Initial tree}\label{sec:initialtree}


In the context of BSTs, an \emph{initial tree} refers to the tree structure existing before executing a sequence. For $\textsc{Greedy}$, an initial tree corresponds to the set of points on the $x,y$ plane before executing a sequence (See \cref{fig:initialvsflat}). 

In the existing literature, typically, for $\textsc{Greedy}$, the initial tree is either \emph{flat} (equivalent to having no points in the plane prior to execution), which we denote as $I_f$ (See Figure 3(b)), or it can be arbitrary (See Figure 3(a)). For a search sequence $S$, the execution of $\textsc{Greedy}$ can differ significantly depending on whether it starts with an arbitrary initial tree or a flat initial tree (See \cref{fig:initialvsflat}). As a result, the cost of $\textsc{Greedy}$ varies for the same search sequence $S$ across different initial trees.  To overcome this problem, researchers sometimes use a pre-processing step to convert an arbitrary initial tree to $I_f$ (See Section 2 in \cite{chalermsook2015pattern}). After obtaining the flat initial tree, the execution of the search sequence $S$ begins. However, in some cases, the Splay tree and $\textsc{Greedy}$ are also studied on specific types of initial trees \cite{chaudhuri1993splaying,chalermsook2015pattern}. But still, there has been limited research on executing a sequence $S$ with a specific initial tree. To address this, we introduce an initial tree for $\textsc{Greedy}$ called the \emph{Permutation initial tree}.

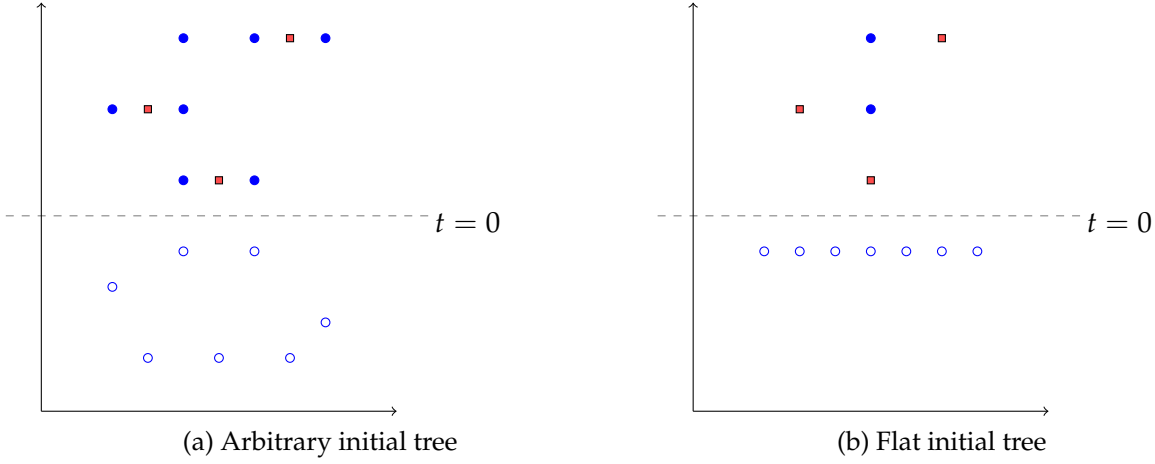
\begin{figure}[hpt!]

\begin{subfigure}{0.47\textwidth}
\hspace{-.2cm}
  \begin{tikzpicture}[scale=.47]
\draw [gray,dashed] (-1,0)--(11,0) ;
\draw[->][black] (0,-5.5) -- (10,-5.5);
\draw[->][black] (0,-5.5) -- (0,6);

\filldraw[draw=black,fill=red!70!white] (4.9,0.9) rectangle (5.1,1.1) (2.9,2.9) rectangle (3.1,3.1) (6.9,4.9) rectangle (7.1,5.1) ;

\filldraw[blue] (4,1) circle (3.5pt) (6,1) circle (3.5pt) (2,3) circle (3.5pt) (4,3) circle (3.5pt) (6,5) circle (3.5pt) (4,5) circle (3.5pt) (8,5) circle (3.5pt);

\filldraw[draw=blue, fill=white] (4,-1) circle (3.5pt) (6,-1) circle (3.5pt) (2,-2) circle (3.5pt) (8,-3) circle (3.5pt) (3,-4) circle (3.5pt) (7,-4) circle (3.5pt) (5,-4) circle (3.5pt);

\draw[black](12,0.4)node[anchor=north]{$t=0$};
\end{tikzpicture}
\captionsetup{justification=centering}
    \caption{Arbitrary initial tree}
    \label{fig:initialtree}
\end{subfigure}
\begin{subfigure}{0.47\textwidth}
\hspace{0.2cm}
   \begin{tikzpicture}[scale=.47]
\draw [gray,dashed] (-1,0)--(11,0) ;
\draw[->][black] (0,-5.5) -- (10,-5.5);
\draw[->][black] (0,-5.5) -- (0,6);

\filldraw[draw=black,fill=red!70!white] (4.9,0.9) rectangle (5.1,1.1) (2.9,2.9) rectangle (3.1,3.1) (6.9,4.9) rectangle (7.1,5.1) ;

\filldraw[blue] (5,3) circle (3.5pt) (5,5) circle (3.5pt);

\filldraw[draw=blue, fill=white] (5,-1) circle (3.5pt) (4,-1) circle (3.5pt) (6,-1) circle (3.5pt) (2,-1) circle (3.5pt) (8,-1) circle (3.5pt) (3,-1) circle (3.5pt) (7,-1) circle (3.5pt);

\draw[black](12,0.4)node[anchor=north]{$t=0$};
\end{tikzpicture}
\captionsetup{justification=centering}
\caption{Flat initial tree}
    \label{fig:flattree}
    
\end{subfigure}
\caption{Execution of $\GR$ on two different initial trees. Hollow points below $t=0$ represent the initial set of points (initial tree). $(a)$ depicts an arbitrary initial tree. $(b)$ depicts a flat (empty) initial tree.}
\label{fig:initialvsflat}
\end{figure}

\begin{definition}\label{def:initialtree}(Permutation Initial Tree) An initial tree denoted as $I_p$ is called a permutation initial tree if the set of points in the plane before the execution of $\textsc{Greedy}$ on a sequence is a permutation, i.e., there are no two points $p_1$ and $p_2$ in the initial tree such that $p_1.x=p_2.x$ and $p_1.y=p_2.y$.
\end{definition}

Let $I_p = (z_1, z_2, \dots, z_n)$ be a permutation initial tree (See  \cref{fig:permutationinitial}), where each $z_i$ represents a point parallel to the $x$-axis at time $-i$. In other words, $z_i$ is a point in the initial tree corresponding to a key with negative time, indicating that $z_i$ exists prior to any search.

\begin{remark}
    We can also view the permutation initial tree $I_p$ as a permutation sequence $S=(z_1, z_2, \dots, z_n)$.
\end{remark}

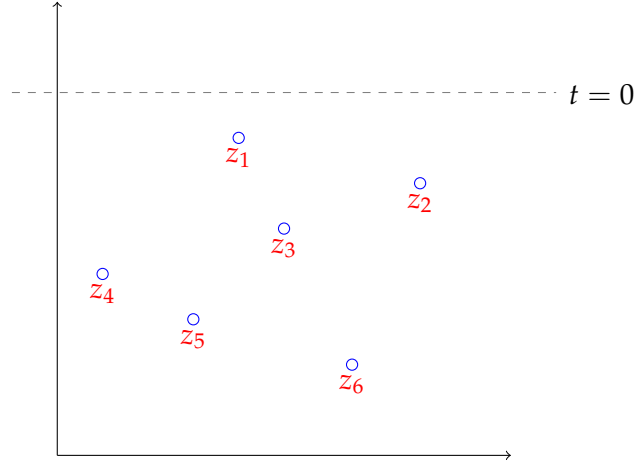
\begin{figure}
\centering
\begin{tikzpicture}[scale=.60]
\draw [gray,dashed] (-1,0)--(11,0) ;
\draw[->][black] (0,-8) -- (10,-8);
\draw[->][black] (0,-8) -- (0,2);



\filldraw[draw=blue, fill=white] (4,-1) circle (3.5pt) (8,-2) circle (3.5pt) (5,-3) circle (3.5pt) (1,-4) circle (3.5pt) (3,-5) circle (3.5pt) (6.5,-6) circle (3.5pt);

\draw[black](12,0.4)node[anchor=north]{$t=0$};
\draw[red] (4,-1)node[anchor=north]{$z_1$} (8,-2)node[anchor=north]{$z_2$} (5,-3)node[anchor=north]{$z_3$} (1,-4)node[anchor=north]{$z_4$} (3,-5)node[anchor=north]{$z_5$} (6.5,-6)node[anchor=north]{$z_6$} ;
\end{tikzpicture}

\captionsetup{justification=centering}
\caption{Example of a permutation initial tree.}
    \label{fig:permutationinitial}

\end{figure}

In the next section, we derive a relation between the execution of $\textsc{Greedy}$ on the flat initial tree and the permutation initial tree.

\section{$\textsc{Greedy}$ on Permutation Initial Tree}\label{sec:GRpermu}

In the previous section, we discussed how the cost of $\textsc{Greedy}$ is different for different initial trees on the same search sequence $S$. With this in mind, we ask the following question: \emph{Does there exist an initial tree for which the cost of $\textsc{Greedy}$ on a sequence \( S \) matches its cost on the flat initial tree?} We answer this question in the affirmative by showing that for any permutation search sequence $S$, there exists a permutation initial tree $I_p \neq I_f$ such that the cost of $\textsc{Greedy}$ to execute the sequence $S$ on $I_f$ and $I_p$ is the same. To this end, we show the following lemma.

\begin{lemma}\label{lem:notouchinitial}
    Let $S$ be a permutation search sequence, and let $I_p$ be the permutation initial tree such that $I_p=\mathcal{M}(S)$. Then $\textsc{Greedy}$ does not touch any key from the initial tree $I_p$ at time $i \in [n]$.
\end{lemma}
\begin{proof}
    Let $s_i \in S$ be a search key at time $i$ and let $I_p$ be the initial tree corresponding to the search sequence $S$, such that $I_p = \mathcal{M}(S)$. We want to show that no keys from the initial tree $I_p$ are touched by $\textsc{Greedy}$ at time $i$. As $I_p=\mathcal{M}(S)$, we know that there is a point $z_i$ below $s_i$ at time $-i$. Any point below $z_i$ in the initial tree $I_p$ forms an arborally satisfied rectangle with $s_i$ at time $i$ because all these rectangles are satisfied due to $z_i$ (See  \cref{fig:pointsatisfied}). Therefore, any point below $z_i$ from the initial tree $I_p$ cannot be touched by $\textsc{Greedy}$ at time $i$.

As $I_p=\mathcal{M}(S)$, we know that any point in the initial tree $I_p$, above $z_i$, has already been searched before time $i$. Therefore, $s_i$ at time $i$ can only form an arborally unsatisfied rectangle with search keys $s_1$ to $s_{i-1}$. Therefore, any point above $z_i$ from the initial tree $I_p$ cannot be touched by $\textsc{Greedy}$ at time $i$. Hence, no points in the initial tree $I_p$ is touched by $\textsc{Greedy}$ at time $i$.
\end{proof}

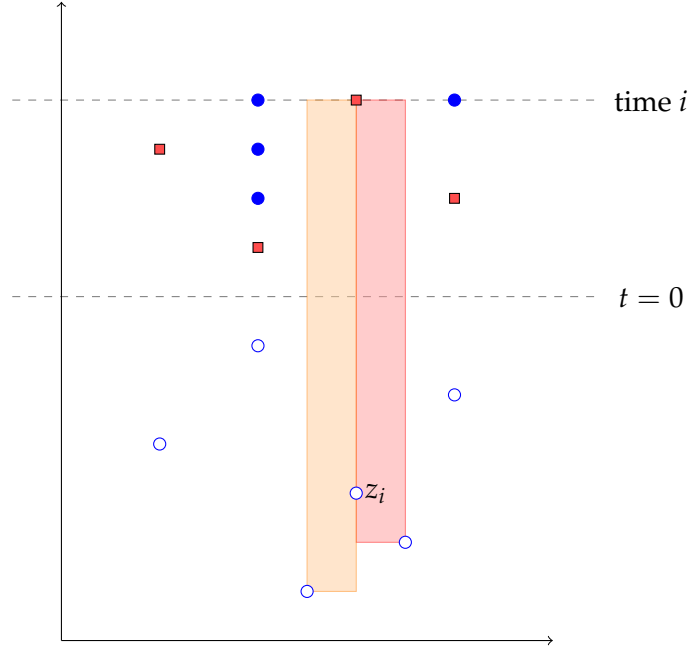
\begin{figure}[hpt!]
\centering
\begin{tikzpicture}[scale=.65]
\draw [gray,dashed] (-1,0)--(11,0) ;
\draw [gray,dashed] (-1,4)--(11,4) ;
\draw[->][black] (0,-7) -- (10,-7);
\draw[->][black] (0,-7) -- (0,6);

\filldraw[red!40!white,opacity=.5,draw=red] (6,4) rectangle (7,-5);
\filldraw[orange!40!white,opacity=.5,draw=orange] (6,4) rectangle (5,-6);

\filldraw[draw=black,fill=red!70!white] (3.9,0.9) rectangle (4.1,1.1) (7.9,1.9) rectangle (8.1,2.1) (1.9,2.9) rectangle (2.1,3.1) (5.9,3.9) rectangle (6.1,4.1) ;

\filldraw[blue] (4,2) circle (3.5pt) (4,3) circle (3.5pt) (4,4) circle (3.5pt) (8,4) circle (3.5pt)  ;

\filldraw[draw=blue, fill=white] (4,-1) circle (3.5pt) (8,-2) circle (3.5pt) (2,-3) circle (3.5pt) (6,-4) circle (3.5pt) (7,-5) circle (3.5pt) (5,-6) circle (3.5pt);

\draw[black](12,0.4)node[anchor=north]{$t=0$};
\draw[black](12,4.4)node[anchor=north]{time $i$};
\draw[black] (6.4,-3.6)node[anchor=north]{$z_i$} ;
\end{tikzpicture}

\caption{Execution of $\textsc{Greedy}$ on an initial tree $I_p=\mathcal{M}(S)$. The two shaded rectangles are arborally satisfied at time $i$ due to the point $z_i$ in the initial tree.}
    \label{fig:pointsatisfied}

\end{figure}

With the above lemma in hand, we now prove the following Theorem.

\begin{theorem}\label{lem:samecost}
    Let $S$ be a permutation search sequence, and let $I_p$ be a permutation initial tree such that $I_p=\mathcal{M}(S)$, then $\textsc{Greedy}_{I_p}(S)=\textsc{Greedy}_{I_f}(S)$.
\end{theorem}
\begin{proof}
Analyzing $\textsc{Greedy}$ on a flat initial tree (or no initial tree) is straightforward. When $\textsc{Greedy}$ performs a search $s_i \in S$ at time $i$, it can only touch the search keys up to time $i-1$. This means $\textsc{Greedy}$ at time $i$ does not touch keys that have not been searched yet. However, if the initial tree was not flat, $\textsc{Greedy}$ might have touched keys that are yet to be searched (See Figure 3(a)). 

 To prove the equivalence that $\textsc{Greedy}_{I_p}(S)=\textsc{Greedy}_{I_f}(S)$, we prove the following two conditions. 
 \begin{itemize}
     \item (a) $\textsc{Greedy}$ at time $i$ touches the same set of keys in both the initial tree i.e., $I_p$ and $I_f$.

     \item (b) $\textsc{Greedy}$ at time $i$ does not touch any key from the initial tree $I_p$.
 \end{itemize} 
 
 As we have already proved the latter condition in  \cref{lem:notouchinitial}, we now prove the first condition.

 For brevity, when we say a key (or a point) is touched in $I_p$ (or $I_f$), we mean a key is touched where the initial tree is $I_p$ (or $I_f$). We use induction for the proof.  When $s_1$ is searched at time $1$, in $I_p$ and $I_f$, $\textsc{Greedy}$ does not touch any point. Suppose by induction, the set of keys touched by $\textsc{Greedy}$ till time $i-1$ in both $I_p$ and $I_f$ are the same.

At time $i$, when $s_i$ is searched, $\textsc{Greedy}$ will find all the arborally unsatisfied rectangles in $I_p$ and $I_f$. The set of keys that can be touched at time $i$ are among $s_1$ to $s_{i-1}$ as we already know that no key from the initial tree is touched at time $i$. We now claim that the keys touched at time $i$ in $I_p$ and $I_f$ are the same. Suppose, for contradiction, that a key $q$ is touched in $I_p$ at time $i$ but not in $I_f$. This implies $\square_{s_iq}$ is not arborally satisfied in $I_p$ but $\square_{s_iq}$ is arborally satisfied in $I_f$. Therefore, there exists a point, say $q'$ in $I_f$ before time $i$ for which $\square_{s_iq}$ is satisfied. But by induction, this point is also present in $I_p$, which contradicts that $\square_{s_iq}$ is arborally unsatisfied in $I_p$. Similarly, we can show a contradiction when a key $q$ is touched in $I_f$ but not in $I_p$. Hence, at time $i$, $\textsc{Greedy}$ touches the same set of keys in both $I_p$ and $I_f$.

\end{proof}

 \cref{lem:samecost} shows that there is an initial tree other than the flat initial tree such that the cost of $\textsc{Greedy}$ on any search sequence $S$ is the same on both the initial trees.

    \begin{corollary}\label{cor:samecost}
    Let $I_p$ be a permutation initial tree, and let $S$ be a permutation sequence such that $S=\mathcal{M}(I_p)$, then $\textsc{Greedy}_{I_p}(S)=\textsc{Greedy}_{I_f}(S)$.
    \end{corollary}

  \begin{corollary}\label{cor:samecost1}
    Let $S$ be a search sequence. Let the $OPT(S)=\textsc{Greedy}_{I_f}(S)$ then $\textsc{Greedy}_{I_p}(S)=OPT(S)$, where $I_p=\mathcal{M}(S)$.
\end{corollary}

A natural question arises regarding the application of \cref{lem:samecost}. In the following section, we demonstrate how \cref{lem:samecost} can be used to derive an \( O(n) \) bound for any preorder sequence \( S \) on a permutation initial tree \( I_p = \mathcal{M}(S) \). While we specifically illustrate this application for preorder sequences, the theorem can similarly be applied to any sequence \( S \) for which the bound $\textsc{Greedy}_{I_f}(S)$ is known.

\section{Preorder on Permutation Initial Tree}\label{sec:preorder}

As discussed earlier, the cost of the preorder sequence on $\textsc{Greedy}$ with a flat initial tree is known to be $O(n)$ \cite{chalermsook2015pattern}. For arbitrary initial tree, in \cite{chalermsook2015pattern}, the authors gave a bound of $n2^{\alpha(n)^{O(1)}}$ which was recently improved to $O(n2^{\alpha(n)})$.

In \cite{chalermsook2015pattern}, the authors showed that any preorder sequence can be executed in $O(n)$ time if preprocessing is permitted. In the preprocessing phase, they convert an arbitrary initial tree into a flat initial tree in $O(n)$ time. Subsequently, they show that the cost of performing preorder traversal on the resulting flat initial tree ($I_f$) is $O(n)$. For completeness, we include their corollary as follows.

\begin{corollary}(Corollary 1.5 of \cite{chalermsook2015pattern})\label{cor:preorderparinya} 

The cost of accessing a preorder sequence using $\textsc{Greedy}$, with preprocessing, is linear.
\end{corollary}

We get the following theorem for the preorder sequence using  \cref{lem:samecost} and  \cref{cor:preorderparinya}.

\begin{theorem}\label{thm:preorder}
  Let $S$ be a preorder sequence, and let $I_p=\mathcal{M}(S)$ then $\textsc{Greedy}_{I_p}(S)=O(n)$.
\end{theorem}

This is the first result showing that $\textsc{Greedy}$ on preorder sequence is $O(n)$ when the initial tree is not a flat tree ($I_f$), matching the similar result obtained for Splay trees by  \cite{chaudhuri1993splaying}.  We now define a permutation initial tree, named a preorder initial tree, as follows.

\begin{definition}(Preorder Initial tree)

     A permutation initial tree that avoids $(2,3,1)$ as a subsequence is called a preorder initial tree. 
\end{definition}
Similar to Corollary \ref{cor:samecost}, we get the following corollary.

\begin{corollary}\label{cor:preorder}
    For every preorder initial tree $I_{pre}$, there is a preorder sequence $S=\mathcal{M}(I_{pre})$ such that $\textsc{Greedy}_{I_{pre}}(S)=O(n)$.
\end{corollary}

   \section{Conclusion and Open Problems}\label{sec:conclusion}
In this paper, we analyzed the cost of $\textsc{Greedy}$ on the permutation initial tree, demonstrating that it matches the cost of $\textsc{Greedy}$ on the flat initial tree for all search sequences \( S \).  Any tight bound for $\textsc{Greedy}$ on the flat initial tree directly implies the same bound on the permutation initial tree. Consequently, an interesting direction for future work is identifying more sequences for which $\textsc{Greedy}$ is optimal when the initial tree is flat. Another promising direction is to explore additional classes of initial trees where the cost of $\textsc{Greedy}$ matches its cost on the flat initial tree.

For the preorder sequences, we showed that $\textsc{Greedy}$ achieves the optimal bound of \( O(n) \) on permutation initial trees.
 An interesting direction would be identifying more initial trees for which $\textsc{Greedy}$ is optimal for preorder sequences. This could provide valuable insights into the preorder conjecture, which is defined for arbitrary initial trees.

\bibliographystyle{alpha}
\newcommand{\etalchar}[1]{$^{#1}$}

\end{document}